\documentclass[conference]{IEEEtran}
\IEEEoverridecommandlockouts
\usepackage{cite}
\usepackage{amsmath,amssymb,amsfonts,amsthm}
\usepackage{algorithmic}
\usepackage{graphicx}
\usepackage{textcomp}
\usepackage{xcolor}
\usepackage{mathtools}
\newtheorem{theorem}{Theorem}
\newtheorem{observation}{Observation}
\newtheorem{lemma}{Lemma}
\theoremstyle{definition}
\newtheorem{definition}{Definition}
    
\newcommand{\sn}{\mathbf{s}^{(n)}}

\newcommand{\ham}{H_W(\mathbf{s}^{(n)} )}
\begin{document}
\title{Thermodynamics as Combinatorics: A Toy Theory}

\author{\IEEEauthorblockN{ \"Amin Baumeler }
	\IEEEauthorblockA{\textit{ IQOQI-Vienna } \\
		\textit{Austrian Academy of Sciences}\\
		Vienna, Austria \\
		 aemin.baumeler@oeaw.ac.at}
	\and
	\IEEEauthorblockN{ Carla Rieger}
	\IEEEauthorblockA{\textit{Department of Physics} \\
		\textit{Eidgen\"ossische Technische Hochschule}\\
		Z\"urich, Switzerland \\
		carieger@ethz.ch}
		\and
	\IEEEauthorblockN{ Stefan Wolf}
	\IEEEauthorblockA{\textit{Faculty of Informatics} \\
		\textit{Universit\`a della Svizzeria italiana}\\
		Lugano, Switzerland \\
		wolfs@usi.ch}

}

\maketitle

\begin{abstract}
We discuss a simple toy model which allows, in a~natural way, for deriving central facts from thermodynamics such as its fundamental laws, including Carnot's version of the second principle.
Our viewpoint represents thermodynamic systems as binary strings, and it links their temperature to their Hamming weight.
From this, we can reproduce the possibility of negative temperatures, the notion of equilibrium as the co\"{\i}ncidence of two notions of temperature~--- statistical versus structural~---, as well as the zeroth law of thermodynamics (transitivity of the thermal-equilibrium relation), which we find to be redundant, as other authors, yet at the same time not to be universally valid.
\end{abstract}

\section{Introduction}
In the present article, we attempt to understand and develop the links between
the well-established physical theory of {\em thermodynamics\/} on the one hand, 
 and {\em information (as well as algorithmic-complexity) theory\/} on the other.
Let us note first that, given the variety of already 
known connections~--- most obviously given through the notion of 
{\em entropy\/}~---, it is not surprising that close ties exist between 
these fields. 

In line of previous work described below, we propose a toy model of thermodynamics
featuring binary strings of infinite length (the ``thermodynamic limit'') as heat baths, 
where the strings' Hamming weight per length (``Hamming fraction'') is linked to their 
temperature. Our model is in the same spirit as a proposal by Spekkens~\cite{spekkens2007evidence},
presenting a toy model able to capture key properties of quantum physics that is based on a hidden-variable model. With this, he defends an epistemic view
of quantum states. Our combinatorial model of thermodynamics, in turn, can capture key features of the physical theory, such as its fundamental laws, suggesting 
that the essence of thermodynamics can be derived through combinatorics, i.e., counting arguments. 

In our model, we reproduce the notions of {\em heat reservoir, temperature, equilibrium,} and {\em entropy.} We
express in a natural and simple way the fundamental laws of thermodynamics 
from the zeroth to the second, including notably the first~(Carnot's) version of the 
second law~\cite{carnot1978reflexions}, bounding the efficiency of a circular process between
two heat baths as one minus the ratio between their absolute temperatures.

We also obtain in a straight-forward way the concept 
of negative temperature. The phenomenon arises for finite strings 
the Hamming fraction of which exceeds $1/2$~--- when the increase of energy ($1$s) 
effectively {\em de\/}creases the number of degrees of freedom~---, 
but also with respect to {\em structural temperature\/}, an alternative measure we define 
based on the notion of Kolmogorov complexity as entropy. 
\begin{figure}[!t]
    \centering
    \includegraphics[width = 0.35\textwidth]{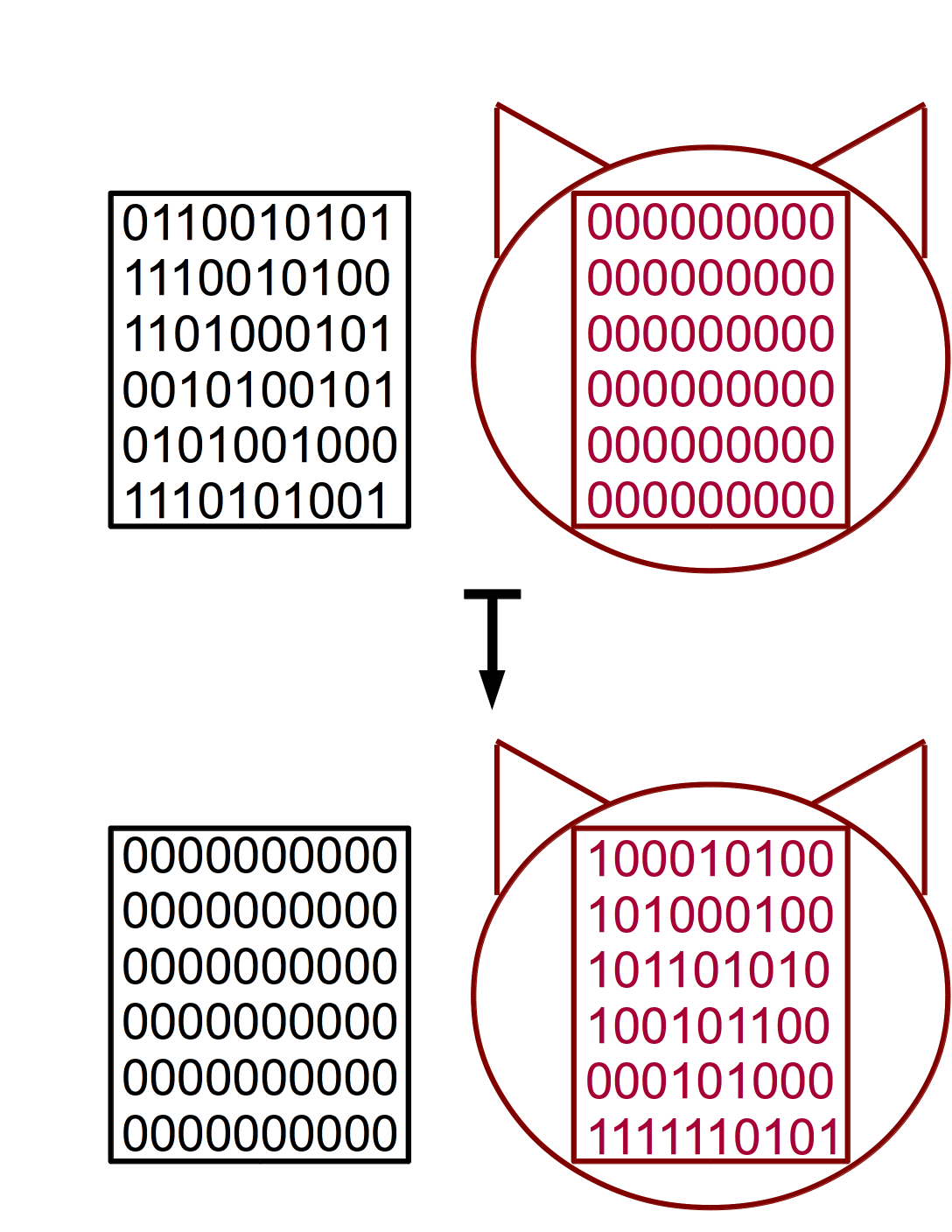}
    \caption{The resolution of the Maxwell-demon paradox as described by Bennett: The order that is created {\em outside\/} the demon is compensated by the disorder arising {\em inside\/} the demon~\cite{bennett1982thermodynamics}.}
    \label{fig:demon}
\end{figure}

\subsection{Previous results}
A close connection between the second law of thermodynamics and information was expressed
by Rolf Landauer~\cite{landauer1961irreversibility}, whose starting point was his famous slogan ``Information is Physical''~\cite{landauer1991}:
The erasure of $N$ bits of information costs at least\footnote{Here, $k_B$ stands for Boltzmann's constant.}~$k_B T\ln 2\cdot N$ of free 
energy, the dissipation of which as heat into the environment compensates for 
the entropy decrease of the erased memory. This principle was used by Bennett~\cite{bennett1982thermodynamics}
to resolve the paradox around ``Maxwell's demon'' (see Figure~\ref{fig:demon}):
The latter, in the course of her gas-sorting activity, accumulates information in her brain, and the corresponding 
growth of entropy compensates for the external entropy defect. The erasure of that 
information in the demon's internal state would in the end require exactly the 
same amount of free energy that is gained by the sorting~--- the paradox 
disappears.

Motivated by Landauer's principle, Bennett~\cite{bennett1982thermodynamics} as well as Fredkin and Toffoli~\cite{fredkin1982conservative} have 
developed a theory of {\em reversible computing\/}, i.e., a computing model which does not require the 
erasure of information. More specifically, Bennett has described a generic procedure to turn any computation
into a reversible one with essentially the same computational efficiency, whereas Fredkin and Toffoli
have presented a model of computation~--- the ``Ballistic computer,'' based on elastic collisions of balls 
on a billard table~--- that allows for carrying out any logically-reversible computation (no loss of information)
in a thermodynamically-reversible way (no heat dissipation into the environment). 

Based on this, and in generalization of Landauer's principle, the second law of thermodynamics has been speculated to take 
the form that {\em time evolutions are logically reversible\/}~\cite{baumeler2017causality}. In a~spirit related to the toy model 
considered in the present article, we have derived from this fact consequences of this second law 
resembling the formulations due to Clausius as well as Kelvin: ``Heat only flows from a hotter to a 
colder reservoir, not the other way around'' and ``No free energy from one heat bath alone'' (see~Figures~\ref{fig:clausius_second_law},~\ref{fig:kelvin_second_law}).

\begin{figure}[!h]
    \centering
    \includegraphics[width = 0.35\textwidth]{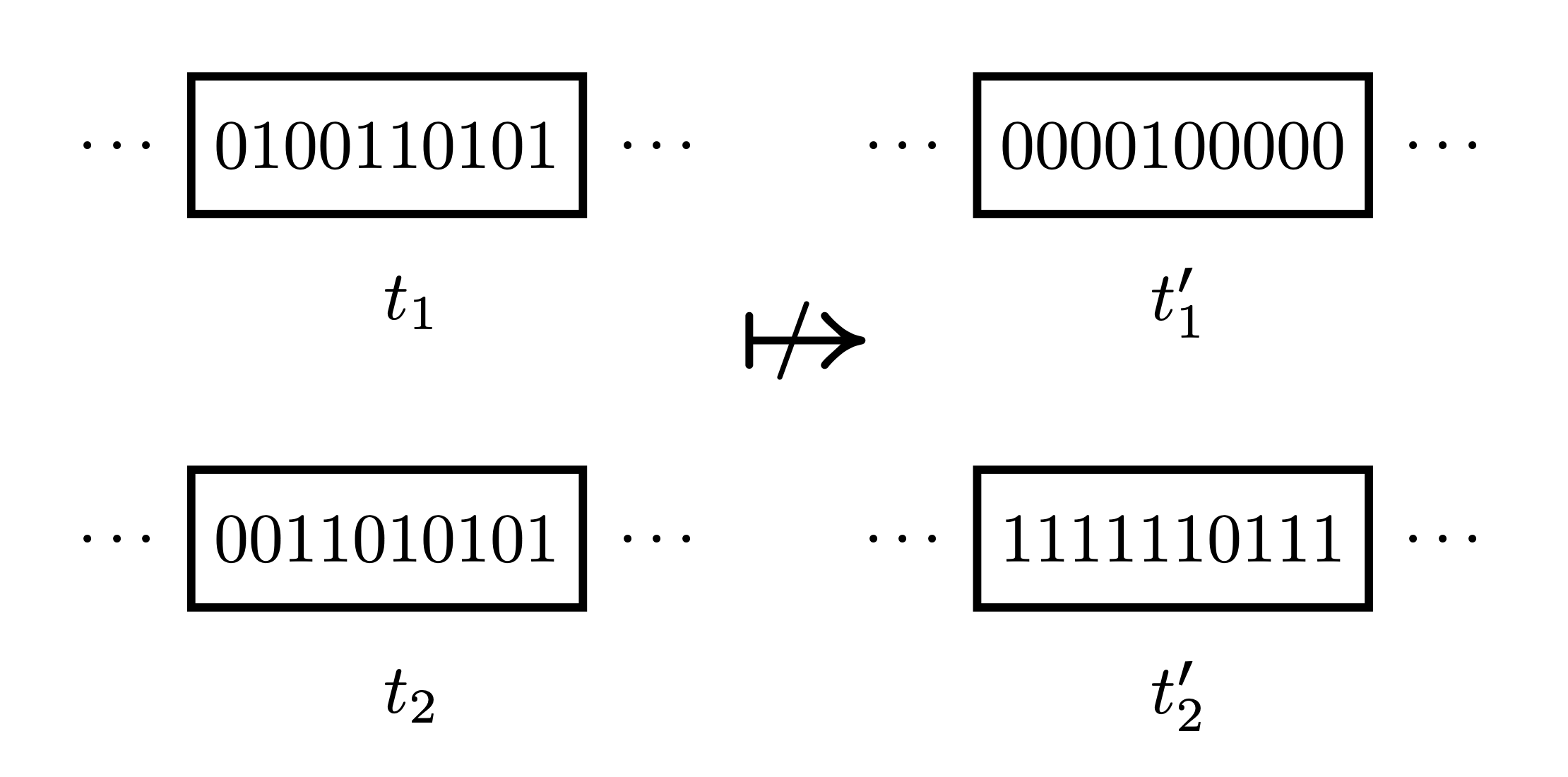}
    \caption{Clausius' second law: Logically reversible maps do not accentuate differences \cite{wolf2018second}, i.e., the depicted transformation is impossible for Hamming fractions $t_1 > t_1'$ and $t_2 < t_2'$ (the total number of $1$s is preserved; this is the first law).}
    \label{fig:clausius_second_law}
\end{figure}

\begin{figure}[!t]
    \centering
    \includegraphics[width = 0.45\textwidth]{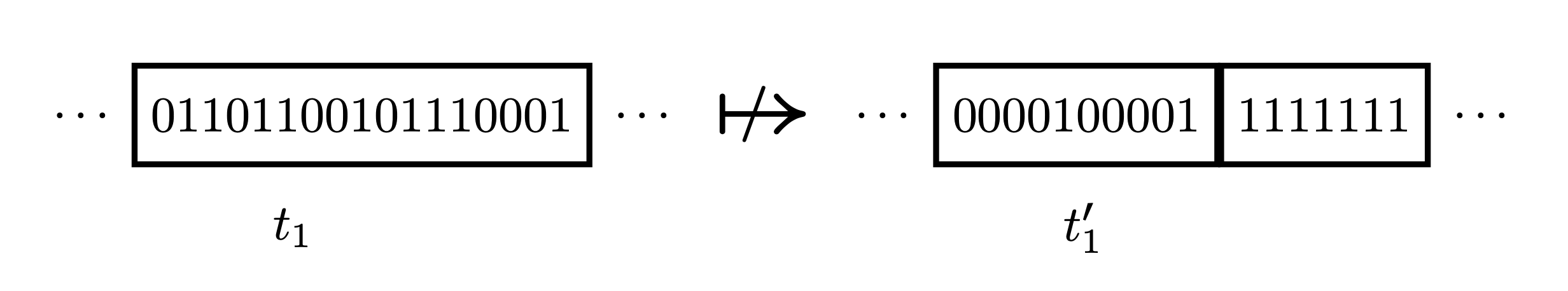}
    \caption{Kelvin's version of the second law: The depicted transformation is impossible for $t_1'<t_1$ (the total number of $1$s is preserved).}
    \label{fig:kelvin_second_law}
\end{figure}

It has been an open question how to reproduce {\em Carnot's\/} version in a similar way. We fill this gap in the present article by presenting a toy model powerful enough to reproduce additional notions such as temperature,
free energy, heat, as well as the efficiency of a~process connecting two heat baths.

The converse of Landauer's principle states that certain, namely {\em redundant\/}, pieces of information have a work value: They allow for transforming environmental heat into free energy. In~\cite{baumeler2017causality} and \cite{baumeler2019free}, general bounds are given, in terms of algorithmic complexity, on the work 
value of a general string as well as the thermodynamic cost of a computation. 

In a recently developed, process-oriented view of thermodynamics~\cite{kammerlander2018zeroth}, the conclusion 
has been drawn that the zeroth law be redundant. We reproduce this conclusion in our toy model presented
in this article: The law is a direct consequence of the notion of thermal equilibrium~--- but it has an error probability: It fails to hold with exponentially small probability in the size of the (random) heat baths. The reason, as we will see, is that the thermal-equilibrium relation fails to be reflexive: Two copies of the same heat bath contain redundancy, thus free energy.

\section{Preliminaries}
First, we introduce the relevant concepts for our toy theory described in Section~\ref{sec:toy_model}. 
The underlying set for our model consists of binary strings~$\mathbf s$ of unbounded length, 
\begin{equation}
   \mathbf{s} = s_1 s_2 s_3 \ldots \, ,
\end{equation} 
where $\sn = s_1 s_2 s_3 \ldots s_n $ is the finite string consisting of the first $n$ bits of $\mathbf{s}$.
\begin{definition}[Hamming fraction]
	The {\em Hamming fraction\/} $t(\mathbf s)$ of a string~$\mathbf s$ is defined in the asymptotic limit as
	\begin{equation}
		t (\mathbf{s}) := \lim_{n\to\infty} \frac{\ham}{n}
		\,,
	\end{equation}
	where~$\ham$ denotes the Hamming weight of string~$\sn$, and $n$ is the length of $\sn$.
\end{definition}
\noindent
We restrict our model to strings for which this limit exists.

\begin{definition}[Kolmogorov complexity]
	The {\em Kolmogorov complexity\/}~$K(\mathbf s)$~\cite{kolmogorov1968three} of a string~$\mathbf{s}$ is the length of the shortest program~$p$ for a universal\footnote{The Kolmogorov complexity $K_{\mathcal{U}}$ depends on the universal Turing machine~$\mathcal{U}$ only up to an additive constant. Since our considerations concern asymptotic limits, we ignore the Turing-machine dependence from now on.} Turing machine~$\mathcal{U}$ generating~$\mathbf{s}$:
	\begin{equation}
		K_{\mathcal{U}}(\sn) :=\min \{ |p| \ | \ \mathcal{U} (p) = \sn \}
		\,.
		\label{eqn:kolmogorov_compl}
	\end{equation}
\end{definition}

\begin{lemma}
	The Kolmogorov complexity $K(\mathbf s)$ is upper bounded by:
	\begin{equation}
	    \frac{K(\sn)}{n} \lesssim h(t)
	    \,,
	\end{equation}
	where $h$ is the binary entropy function.
\end{lemma}

\section{The Toy Model: Strings and Their Temperature} \label{sec:toy_model}
The {\em macrostate\/} of a system~$\mathbf s$ is characterized by its length~$n$ and its Hamming fraction~$t$.
The number of possible strings~(i.e., {\em microstates\/} for given macroscopic observables) with fixed $n$ and $t$ is given by 
\begin{equation}
   \binom{n}{t \cdot n} \approx 2^{n\cdot h(t)} = 2^{n\cdot (-t\cdot \log_2(t) - (1-t) \cdot \log_2(1-t))}\,.
\end{equation}

We define below the (statistical) temperature of a string depending on its Hamming fraction. We first review an element of classical Boltzmann theory from statistical mechanics in order to motivate our definition, showing that the latter is in accordance with the traditional view. The probability to find a particle with energy $E_i$ in a system at temperature $T$ with total particle number $N = \sum_i N_i$ and total energy $E = \sum_i n_i \cdot E_i$, in the case that $N_i$ particles are occupying the energy state~$E_i$ for $i \in \{1,\ldots,J \}$, is given by the Boltzmann equilibrium distribution~\cite{hardy2014thermodynamics}:
\begin{equation}
    \frac{N_i}{N} = \frac{1}{Z} e^{-E_i/k_B T}\,.
    \label{thermal_state}
\end{equation}
This expression includes the classical partition function~\mbox{$Z = \sum_i e^{-E_i/k_B T}$}.
The formula~(\ref{thermal_state}) was derived by Boltzmann in a~{\em combinatorial\/} approach \cite{sep-statphys-Boltzmann}, maximizing the number of ways of arranging $N$ particles on energy states $E_i$ with occupation numbers $N_i$, while keeping $E$ and $N$ conserved.
For the binary string~\mbox{$\sn \in \{0,1 \}^n$}, each bit can be merely in two different states, i.e., the energy state $0$ or $\Delta E$, and based on (\ref{thermal_state}), one obtains the following probability for a bit $s_i$ being equal to $1$ at energy~$\Delta E$:
\begin{equation}
    P(s_i=1)\approx \frac{\ham}{n} = \frac{e^{- \Delta E/k_B T}}{1+e^{- \Delta E/k_B T}} \text{ for } i \in \{1, \ldots, n\}\,.
    \label{eqn:boltz}
\end{equation}
 Thereby, for each of the two-level systems, solely the energy difference $\Delta E$ of both states is relevant. If we consider~$\Delta E$ as a constant for a given string $\sn$, we may redefine the temperature $T' := k_B T/\Delta E$, and we set $P(s_i = 1)$ equal to the Hamming fraction $t$. If we solve (\ref{eqn:boltz}) for $T'$, then we obtain a term for defining the following temperature.

\begin{definition}[Statistical Temperature]
	We define the {\em (statistical) temperature\/} $T_{\text{stat}}$ of a string $\mathbf{s}$ as
	\begin{equation}
		\frac{1}{T_{\text{stat}}(\mathbf{s})} \coloneqq \log_2 \left( \frac{1-t(\mathbf{s})}{t(\mathbf{s})} \right)\,.
	    \label{def:stat_temp}
	\end{equation}
\end{definition}
\noindent
The dependence of the temperature $T_{\text{stat}}$ on $t$ is shown in Figure~\ref{fig:temperature}.

\section{Energy and the First Law}
The binary string represents an encoding of energy fractions for each degree of freedom of a macroscopic system, i.e., each element of the binary string $s_i$ represents a degree of freedom that either carries energy or not. Hence, we may think of the~$1$s as representations of units of energy: Each (binary) element is either in its ground state $(0)$ or {\em upper\/}
energy state $(1)$. Thus, the total number of $1$s, i.e., the total Hamming weight, encodes the total energy contained within the specific system.

\begin{figure}
    \centering
    \includegraphics[width = 0.4\textwidth]{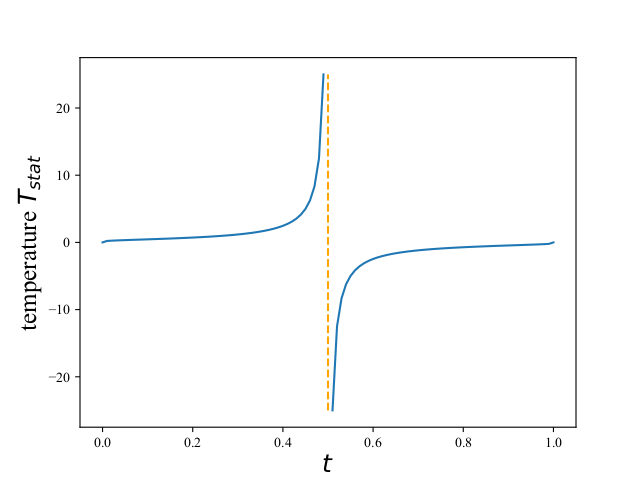}
    \caption{The temperature $T_{\text{stat}}$ in dependence on the Hamming fraction~$t$ of a~string~$\mathbf{s}$. }
    \label{fig:temperature}
\end{figure}

\begin{definition}[First Law]
    A transformation applied to a set of finite strings is said to respect the first law if it preserves the total Hamming weight of the strings.
\end{definition}

\section{Carnot's Version of the Second Law}
In accordance with previous related models, we consider the second law of thermodynamics as the logical reversibility of a transformation.
\begin{definition}[Second Law]
    A physical transformation acting on a set of finite strings is said to satisfy the second law of thermodynamics if it is logically reversible, i.e., no information is lost through the transformation.
\end{definition}
From this fact, it has already been observed that Clausius- and Kelvin-like versions of the second law follow. It was an open question whether the same holds for Carnot's theorem~--- which we show here. In order to achieve this, we consider physical transformation respecting {\em both\/} the first and the second law:
A physical transformation on a~finite string follows the first {\em and\/} second laws if it preserves the number of $1$s {\em and \/} is logically reversible. For such transformations operating on a~pair of strings, we now derive an upper bound on the efficiency of a process extracting free work.

Clausius' version forbids the flow of information from a~colder to a hotter string. We now consider the {\em inverse\/} scenario: If heat (a certain number of $1$s) flows from a hotter to a colder string, then not all the $1$s leaving the hotter string are required to be transferred to the colder one in order to enable the transformation to be logically reversible.~(A~logically reversible map cannot go from a larger to a smaller set.) The fraction of $1$s {\em not\/} required for this compensation can be extracted as free energy from the process. This fraction, when compared to the total number of $1$s leaving the hot string, is defined as the efficiency of the process.

\begin{theorem}[Carnot's Theorem]
	The efficiency~$\eta$ of the Carnot process is
	\begin{equation}
	    \eta = 1-\frac{T_{\text{\rm stat},2}}{T_{\text{\rm stat},1}} 
	\,,
	\end{equation}
	where~$T_{\text{\rm stat},1}$ is the temperature of the hot and~$T_{\text{\rm stat},2}$ of the cold string.
\end{theorem}
\begin{proof}
	Let us consider a string of Hamming fraction~$t_1$ and a colder one with $t_2$ ($<t_1$) being transformed by a heat flow to a lower fraction $ t_1' = t_1 - \Delta_1$ (where $\Delta_1$ is the fraction transferred away from the string), and for the second string a~higher one $t_2'= t_2 + \Delta_2$,~(see Figure~\ref{fig:carnot_second_law}).
	It will turn out that the number of~$1$s to be added to the colder string~$(\Delta_2 \cdot n)$ can be smaller than the number of~\mbox{$1$s~$(\Delta_1 \cdot n)$} taken out of the hotter, and the product of possible strings after the process is still equal to the one before.
	Then, the free energy that can be gained from the process is given by~$(\Delta_1 - \Delta_2) \cdot n$: This, again, is the number of $1$s {\em not\/} required to be transferred to the second string in order to allow for logical reversibility, i.e., guaranteeing that the number of pairs of strings before and after does not get smaller. Specifically,
\begin{equation}
    2^{n \cdot h(t_1)} \cdot 2^{n \cdot h(t_2)} =2^{n \cdot h(t_1 - \Delta_1)} \cdot 2^{n \cdot h(t_2 + \Delta_2)}
    \,.
    \label{eqn:log_revers}
\end{equation}
Since we assume the number of transferred $1$s to be much smaller than the total number of $1$s in the strings, the terms in~(\ref{eqn:log_revers}) can be approximated linearly in $\Delta_1, \Delta_2$, as follows
\begin{equation}
    2^{n \cdot h(t + \Delta)} = 2^{n \cdot h(t)} + (2^{n \cdot h(t)})' \cdot \Delta + \mathcal{O} (\Delta^2)\,,
\end{equation}
where $(\cdot)'$ represents the derivative for $t$. For~(\ref{eqn:log_revers}) to be fulfilled, we must have
\begin{equation}
     \Delta_2 \cdot (2^{n \cdot h(t_2)} ) ' \cdot 2^{n \cdot h(t_1)} = \Delta_1 \cdot (2^{n \cdot h(t_1)} ) ' \cdot 2^{n \cdot h(t_2)}\,.
    \label{eqn:step2}
\end{equation}
Hence, the ratio $\Delta_2/\Delta_1$ is given as
\begin{equation}
    \frac{\Delta_2}{\Delta_1} = \frac{(2^{n \cdot h(t_1)} ) '}{(2^{n \cdot h(t_2)} ) '} \cdot \frac{2^{n \cdot h(t_2)}}{2^{n \cdot h(t_1)}} = \frac{h'(t_1)}{h'(t_2)}\,.
    \label{eqn:ratio}
\end{equation}
Note that the derivative of the binary-entropy function is
\begin{align}
	h'(t) = \log_2\left(\frac{1-t}{t}\right)
\end{align}
and, therefore, through the definition of the statistical temperature, the ratio~$\Delta_2/\Delta_1$ is~$T_{\text{stat},2}/T_{\text{stat},1}$.
Thus, the efficiency $\eta$ of the process is
\begin{equation}
	\eta = \frac{\Delta_1 - \Delta_2}{\Delta_1} = 1- \frac{\Delta_2}{\Delta_1} = 1 - \frac{T_{\text{stat}, 2}}{T_{\text{stat}, 1}}\,.
\end{equation}
\end{proof}
\begin{figure}[!t]
    \centering
    \includegraphics[width = 0.5\textwidth]{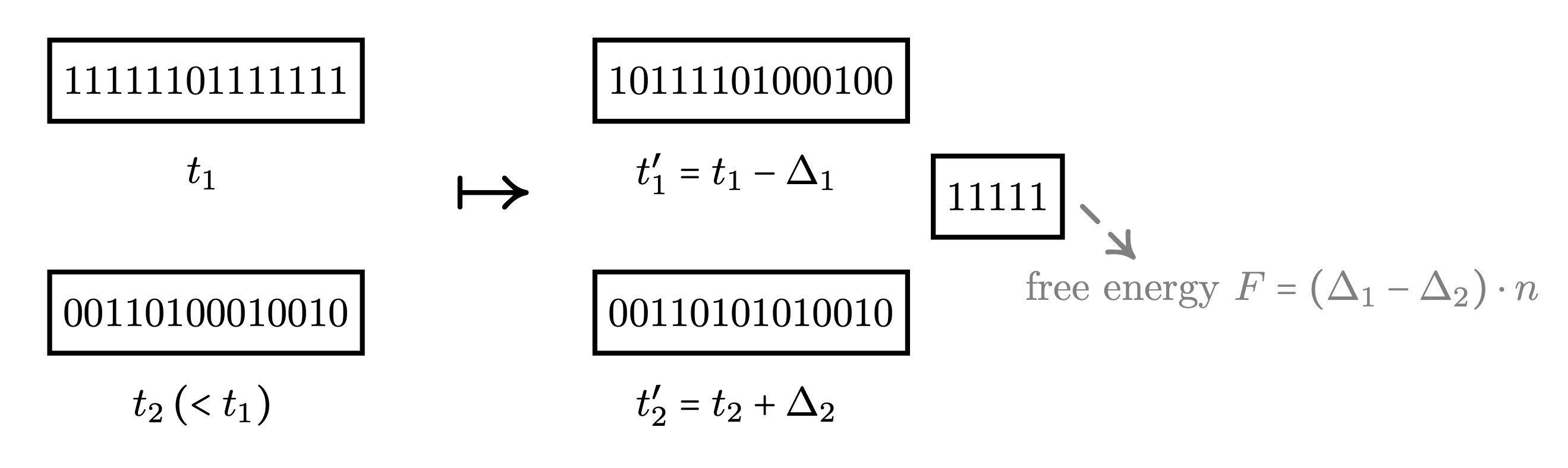}
    \caption{The scenario of Carnot's version of the second law. }
    \label{fig:carnot_second_law}
\end{figure}
The proven fact, which is formally equal to the Carnot efficiency can serve as a confirmation for both the temperature definition based on the Hamming fraction, as well as our reading of the second law as logical reversibility. In this sense, our toy model is consistent and reproduces traditional facts from thermodynamics. Encouraged by this observation, we extend the model to include the notion of entropy.

\section{Complexity as Entropy and Structural Temperature}
We defined the statistical temperature from the Hamming fraction through the derivative of the binary entropy function
\begin{equation}
	\frac{1}{T_{\text{stat}}} = \log_2 \left( \frac{1-t}{t} \right) = h'(t)\,.
\end{equation}
In consequence, the information-theoretic entropy co\"incides with the physical entropy.
\begin{observation}[Entropy]
	The change of the information-theoretic entropy~$\Delta h(t)\cdot n$ equals the change of the physical entropy~$\Delta S$.
\end{observation}
The derivative~$h'(t)$ is the limit of the ratio of the change $\Delta h$ of $h$ given a change $\Delta t$ of $t$, and therefore, the inverse statistical temperature is
\begin{equation}
	\frac{1}{T_{\text{stat}}} \approx \frac{\Delta h(t)}{\Delta t}\,.
\end{equation}
This is consistent with Clausius' entropy definition 
\begin{equation}
    \Delta S = \frac{\Delta Q }{T}
\end{equation}
or, equivalently, 
\begin{equation}
    \frac{1}{T} = \frac{\Delta S}{\Delta Q}.
\end{equation}
Thus, if $\Delta Q = \Delta t \cdot n$ (which is in line with viewing the $1$s as the energy), then $\Delta h(t) \cdot n $ corresponds to the entropy change, making the information-theoretic entropy~$h(t) \cdot n $ of the string correspond
to the physical entropy in our toy model.

However, the actual information content of the string may be smaller than this if it has additional structure, namely,
if the Kolmogorov complexity of the string is smaller than the entropy as defined simply by the fractions of $0$s and $1$s of the string. This serves as a motivation to define a second temperature measure, namely the {\em structural temperature}:
\begin{definition}[Structural Temperature]
	We define the {\em structural temperature\/}~$T_{\text{struc}}$ through
	\begin{equation}
		\frac{1}{T_{\text{struc}}(\mathbf{s})} \coloneqq \max_{\text{position  choices}} \left( \frac{\Delta K (\mathbf{s}) }{\Delta t} \right) \,.
	\end{equation}
	\label{def:struc_temp}
\end{definition}
In contrast to the Hamming fraction, the Kolmogorov-complexity difference depends on the particular transformation, i.e., not merely on the {\em number \/} of $1$s added to the string, but the specific {\em positions\/}. We maximize over all possible choices of these positions. The defined notion of structural temperature allows for assigning a temperature to {\em any\/} string, without ignoring its structure beyond the Hamming fraction. An extremal example of a gap between the two temperatures is the alternating string 
\[ 010101010101010101 \ldots, \] which is statistically hot but structurally cold: Flipping a~bit at position~$k$ induces an arbitrarily large change~$\Delta K(\mathbf s)$ of~$\log_2(k)$.
Such a gap indicates that the string is far from thermal equilibrium. At the other end of the scale, the two temperature notions co\"{\i}ncide exactly for those strings which have no structure besides the one given by the Hamming fraction, i.e., for which the Kolmogorov complexity equals the information-theoretic entropy: 
\begin{equation}
    K(\sn) \approx n \cdot h(t)\,.
\end{equation}
In these latter strings, all the energy is contained in form of heat; and it is exactly those we understand as being in thermal equilibrium or, in other words, as heat baths. Whereas we have, generally 
\begin{align}
T_{\text{struc}} (\mathbf{s}) \leq T_{\text{stat}} (\mathbf{s}) \, ,
\end{align}
heat baths are defined by the equality of the two temperatures.
\begin{definition}[Thermal Equilibrium] \label{def:thermal_equilibirium}
	A string $\mathbf{s}$ is in {\em thermal equilibrium\/} if its statistical temperature equals its structural temperature, i.e.,
	\begin{equation}
		T_{\text{stat}} (\mathbf{s}) =T_{\text{struc}} (\mathbf{s})\,.
	\end{equation}
	Equivalently, we call such a string a {\em heat bath.}
\end{definition}
The notion of thermal equilibrium allows us to talk about the zeroth law of thermodynamics. 

\section{The Zeroth Law}

A string is in thermal equilibrium and, therefore, a heat bath, if its statistical and structural temperature co\"{i}ncide.
This definition extends naturally to {\em pairs\/} of heat baths: By~$\sn \equiv_{t.e.} \mathbf{s}'^{(n)}$, we mean:
\begin{equation}
    K(\sn || \mathbf{s}'^{(n)}) = h \left( \frac{t+ t'}{2} \right) \cdot 2n\,.
    \label{eqn:concat_te}
\end{equation}
\begin{definition}[Zeroth Law]
	The relation $\equiv_{t.e.}$ of thermal equilibrium between pairs of heat baths is {\em transitive.} 
\end{definition}
Clearly, heat baths of different temperatures are {\em not\/}
in thermal equilibrium. On the other hand, and this is the more surprising part, heat baths of the {\em same\/} temperature can fail to be in thermal equilibrium if 
\begin{equation}
    K(\sn || \mathbf{s}'^{(n)}) < K(\sn) + K( \mathbf{s}'^{(n)})\,.
\end{equation}
A consequence thereof is, first, that a heat bath is {\em not\/}
in thermal equilibrium with an identical copy of itself (the relation is irreflexive) since:
\begin{equation}
    K(\sn || \sn) \approx K(\sn).
\end{equation}
Furthermore, this irreflexivity of the relation implies that the zeroth law (transitivity) can also be violated:
Assume that $\mathbf{s}$ and $\mathbf{s}'$ are heat baths that are in thermal equilibrium, then,
\begin{align}
    \mathbf{s}\equiv_{t.e.}\mathbf{s'}\mbox{\ and\ }\mathbf{s'}\equiv_{t.e.}\mathbf{s}\, ,\  \text{but}\ \  \mathbf{s} \not \equiv_{t.e.}\mathbf{s}\ :
\end{align}
The zeroth law, which would imply that $\mathbf{s}$ is also in equilibrium with an identical copy of itself, is violated.
For {\em random\/} heat baths, however, the law holds except with probability exponentially small in their size; in particular, the failure probability vanishes for infinitely large baths. 
In summary, the zeroth law suffers in our model from  exactly the same deficit as the second law when stated as ``entropy always increases:'' It fails to hold with exponentially small probability.

\section{Negative Temperature}
Several authors have already proposed the concept of negative absolute temperature (see, e.g.,~\cite{ramsey1956thermodynamics}). This is consistent with our model. The dependence we propose of the temperature of a string and its Hamming fraction suggests that negative absolute temperatures occur for a string with Hamming fraction $> 1/2$ (see Figure~\ref{fig:temperature}, note that in some sense, strings with negative temperatures are {\em hotter\/}
than those with positive ones).
Intuitively, this is due to the fact that increasing the energy \textit{reduces} the degrees of freedom in terms of the number of possible strings. In our model, also negative {\em structural\/}
temperatures can occur if the string's structure is such that adding energy~--- in the form of $1$s~--- unavoidably makes the string more structured, i.e., its Kolmogorov complexity smaller.

\section{Conclusion}

We have proposed a toy model, based on binary strings, their Hamming weight linked
to the string's ``temperature,'' and have managed to reproduce basic facts 
of thermodynamics, such as its fundamental laws. With this, we
suggest that thermodynamic concepts such as free energy, heat, and thermal 
equilibrium are of combinatorial and computational nature. 
In our model, the notion of {\em negative temperature\/} emerges naturally. First, with respect to statistical temperature: When the Hamming fraction of a string exceeds~$1/2$, then additional energy ($1$s) mean {\em reduction\/} of degrees of freedom. Secondly, also the structural temperature, based on algorithmic complexity, can be negative (if additional energy, i.e., $1$s, reduces complexity). Note that equality of the two temperature measures co\"{\i}ncides with our notion of thermal equilibrium, and is related to the zeroth law, which appears in our model as redundant, underlining earlier observations by other authors. It is, however, valid only except with an error probability exponentially small in the size of the heat baths. This observation, which relates the zeroth law to the second, allowing for similar exceptions, is new.

In contrast to the zeroth and the second law, the situation around the third law
is less clear. We propose as an open question to put its different formulations 
in context with our model in order to get more insight on their relation and 
correctness.

\section*{Acknowledgment}
We thank Xavier Coiteux-Roy and Charles Bédard for interesting discussion on the topic of this article.
This work was supported by the Swiss National Science Foundation (SNF).
\"AB~acknowledges support from the Austrian Science Fund (FWF) through ZK3 (Zukunftskolleg) and F7103 (BeyondC).

\bibliographystyle{IEEEtran}
\bibliography{IEEEabrv,mybibfile}

\begin{thebibliography}{10}
\providecommand{\url}[1]{#1}
\csname url@samestyle\endcsname
\providecommand{\newblock}{\relax}
\providecommand{\bibinfo}[2]{#2}
\providecommand{\BIBentrySTDinterwordspacing}{\spaceskip=0pt\relax}
\providecommand{\BIBentryALTinterwordstretchfactor}{4}
\providecommand{\BIBentryALTinterwordspacing}{\spaceskip=\fontdimen2\font plus
\BIBentryALTinterwordstretchfactor\fontdimen3\font minus
  \fontdimen4\font\relax}
\providecommand{\BIBforeignlanguage}[2]{{%
\expandafter\ifx\csname l@#1\endcsname\relax
\typeout{** WARNING: IEEEtran.bst: No hyphenation pattern has been}%
\typeout{** loaded for the language `#1'. Using the pattern for}%
\typeout{** the default language instead.}%
\else
\language=\csname l@#1\endcsname
\fi
#2}}
\providecommand{\BIBdecl}{\relax}
\BIBdecl

\bibitem{spekkens2007evidence}
R.~W. Spekkens, ``Evidence for the epistemic view of quantum states: A toy
  theory,'' \emph{Physical Review A}, vol.~75, no.~3, p. 032110, 2007.

\bibitem{carnot1978reflexions}
S.~Carnot, ``R{\'e}flexions sur la puissance motrice du feu et sur les machines
  propres {\`a} d{\'e}velopper cette puissance,'' in \emph{Annales
  scientifiques de l'{\'E}cole Normale Sup{\'e}rieure}, vol.~1, 1872, pp.
  393--457.

\bibitem{bennett1982thermodynamics}
C.~H. Bennett, ``The thermodynamics of computation—a review,''
  \emph{International Journal of Theoretical Physics}, vol.~21, no.~12, pp.
  905--940, 1982.

\bibitem{landauer1961irreversibility}
R.~Landauer, ``Irreversibility and heat generation in the computing process,''
  \emph{IBM journal of research and development}, vol.~5, no.~3, pp. 183--191,
  1961.

\bibitem{landauer1991}
------, ``{Information is Physical},'' \emph{Physics Today}, vol.~44, no.~5,
  pp. 23--29, May 1991.

\bibitem{fredkin1982conservative}
E.~Fredkin and T.~Toffoli, ``Conservative logic,'' \emph{International Journal
  of theoretical physics}, vol.~21, no.~3, pp. 219--253, 1982.

\bibitem{baumeler2017causality}
{\"A}.~Baumeler and S.~Wolf, ``Causality--complexity--consistency: Can
  space-time be based on logic and computation?'' in \emph{Time in
  Physics}.\hskip 1em plus 0.5em minus 0.4em\relax Springer, 2017, pp. 69--101.

\bibitem{wolf2018second}
S.~Wolf, ``Second thoughts on the second law,'' in \emph{Adventures Between
  Lower Bounds and Higher Altitudes}.\hskip 1em plus 0.5em minus 0.4em\relax
  Springer, 2018, pp. 463--476.

\bibitem{baumeler2019free}
{\"A}.~Baumeler and S.~Wolf, ``Free energy of a general computation,''
  \emph{Physical Review E}, vol. 100, no.~5, p. 052115, 2019.

\bibitem{kammerlander2018zeroth}
P.~Kammerlander and R.~Renner, ``The zeroth law of thermodynamics is
  redundant,'' \emph{arXiv preprint arXiv:1804.09726}, 2018.

\bibitem{kolmogorov1968three}
A.~N. Kolmogorov, ``Three approaches to the quantitative definition of
  information,'' \emph{International journal of computer mathematics}, vol.~2,
  no. 1-4, pp. 157--168, 1968.

\bibitem{hardy2014thermodynamics}
R.~J. Hardy and C.~Binek, \emph{Thermodynamics and Statistical Mechanics: An
  Integrated Approach}.\hskip 1em plus 0.5em minus 0.4em\relax John Wiley \&
  Sons, 2014.

\bibitem{sep-statphys-Boltzmann}
J.~Uffink, ``{Boltzmann’s Work in Statistical Physics},'' in \emph{The
  {Stanford} Encyclopedia of Philosophy}, {Summer} 2022~ed., E.~N. Zalta,
  Ed.\hskip 1em plus 0.5em minus 0.4em\relax Metaphysics Research Lab, Stanford
  University, 2022.

\bibitem{ramsey1956thermodynamics}
N.~F. Ramsey, ``Thermodynamics and statistical mechanics at negative absolute
  temperatures,'' \emph{Physical Review}, vol. 103, no.~1, p.~20, 1956.

\end{thebibliography}
\end{document}